\documentclass[a4paper, 10pt, conference]{ieeeconf}      % Use this line for a4 paper

\IEEEoverridecommandlockouts                              % This command is only needed if 
\usepackage[margin=0.75in]{geometry}
\newgeometry{top=1in, left=0.75in,right=0.75in,bottom=0.75in}
\usepackage{graphics} % for pdf, bitmapped graphics files
\usepackage{epsfig} % for postscript graphics files
\usepackage{mathptmx} % assumes new font selection scheme installed
\usepackage{times} % assumes new font selection scheme installed
\usepackage{amsmath} % assumes amsmath package installed
\usepackage{amssymb}  % assumes amsmath package installed
\usepackage{algorithm}
\usepackage{algorithmic}
\makeatletter
\def\BState{\State\hskip-\ALG@thistlm}
\makeatother
\usepackage{mathtools}

\DeclareMathOperator*{\argmax}{arg\,max}
\newtheorem{proposition}{Proposition}
\title{An Online Pricing Mechanism for Electric Vehicle Parking Assignment and Charge Scheduling}

\author{Nathaniel Tucker $\quad$ Bryce Ferguson $\quad$ Mahnoosh Alizadeh\\% <-this % stops a space
\thanks{This work was supported by the California Energy Commission from an award administered through SLAC National Accelerator Laboratory. Solicitation: GFO-16-303. Agreement: EPC-16-057.}% <-this % stops a space
\thanks{N. Tucker, B. Ferguson, and M. Alizadeh are with the Department of Electrical and Computer Engineering, University of California, Santa Barbara. \texttt{nathaniel\char`_tucker@ucsb.edu.}}%
}

\begin{document}

\maketitle
\thispagestyle{empty}
\pagestyle{empty}

%%%%%%%%%%%%%%%%%%%%%%%%%%%%%%%%%%%%%%%%%%%%%%%%%%%%%%%%%%%%%%%%%%%%%%%%%%%%%%%%
\begin{abstract}

In this paper, we design a pricing framework for online electric vehicle (EV) parking assignment and charge scheduling. Here, users with electric vehicles want to park and charge at electric-vehicle-supply-equipment (EVSEs) at different locations and arrive/depart throughout the day. The goal is to assign and schedule users to the available EVSEs while maximizing user utility and minimizing operational costs. Our formulation can accommodate multiple locations, limited resources, operational costs, as well as variable arrival patterns. With this formulation, the parking facility management can optimize for behind-the-meter solar integration and reduce costs due to procuring electricity from the grid. We use an online pricing mechanism to approximate the EVSE reservation problem's solution and we analyze the performance compared to the offline solution. Our numerical simulation validates the performance of the EVSE reservation system in a downtown area with multiple parking locations equipped with EVSEs. \\ \hfill

\end{abstract}

%%%%%%%%%%%%%%%%%%%%%%%%%%%%%%%%%%%%%%%%%%%%%%%%%%%%%%%%%%%%%%%%%%%%%%%%%%%%%%%%
\section{INTRODUCTION}

\label{section: intro}

Owners of electric vehicles (EVs) spend a large portion of their average day at work and at home; however, an overlooked third category also contributes to an EV owner's day. The average EV owner spends $1.5 - 4$ hours per day at locations such as shopping centers, travel stops, and restaurants, meaning there is potential for EV charging at these locations \cite{workplace}. However, these miscellaneous locations have highly variable statistics such as arrival time, departure time, and energy requirement that create challenges for EV assignment and charge scheduling. In this paper, we present an EV parking assignment and charge scheduling framework that does not need accurate input statistics with the purpose of increasing smart-charging opportunities at various locations.

A number of past studies have proposed online mechanisms for assigning EVs to electric-vehicle-supply-equipment (EVSEs) as well as scheduling EV charging in geographically limited areas such as parking lots or neighborhoods. Paper \cite{Robu} presents an online mechanism for EV charging with electricity as an expiring resource; however, they allow cancellation of previously allocated resources. Similarly, \cite{Zheng} presents an online algorithm for EV charge scheduling with revocation, meaning that allocations can be cancelled in order to serve new requests with higher valuations. In \cite{auc2charge}, the authors present a budget scaling online auction framework that allows users to update their bids while parked. Paper \cite{Stein} uses a modified consensus algorithm to share electricity between EVs while considering electricity as a perishable and continuously produced resource. Furthermore, \cite{Honarmand2014OptimalSO} considers an intelligent parking lot that maximizes the charge rate of all EVs while taking into account time-varying electricity prices. Paper \cite{bryce} considers a workplace parking structure and presents centralized assignment heuristics  for  EVs  to  EVSEs. Papers \cite{driz_cdc, menu, differentiated, driz_arxiv} propose alternative pricing schemes for EV charging in various related settings. In our previous work \cite{ntucker_allerton}, posted pricing mechanisms were examined for assigning EVs to EVSEs with the goal of maximizing smart charging. However, all users of the system were assumed to arrive in a small time interval every morning instead of arriving in an online fashion throughout the day. 

With the exception of \cite{bryce,ntucker_allerton,SOMC}, most previous work focuses on traditional single-output-single-cable (SOSC) EVSEs which result in large fractions of potential charging time spent idle. To increase smart-charging capabilities and user service, we assume that the destinations are equipped with single-output-multiple-cable (SOMC) EVSEs that can be connected to multiple EVs but only charge one EV at a time \cite{SOMC}. These SOMC EVSEs enable facility management to devise a smart charging plan each day, while satisfying the charging needs of all EVs.

In this paper, we present a framework for the online EVSE reservation problem that accounts for users arriving and departing throughout the day. Our formulation can accommodate many locations, limited resources, operational costs, as well as variable arrival patterns. Additionally, our framework does not revoke previous allocations. With this problem formulation, the parking facility management can optimize for behind-the-meter solar integration and reduce costs due to procuring electricity from the grid. We use an online pricing mechanism to approximate the EVSE reservation problem's solution and we analyze the performance compared to the offline solution. %Additional pricing mechanism results have been submitted in a journal paper \cite{ntucker_TSG}.

The remainder of the paper is organized as follows. Section \ref{section: prob_form} presents the system structure and offline formulation for the EVSE reservation problem. Section \ref{section: Online_Mech} presents the online mechanism used to provide an approximate solution to the EVSE reservation problem and discusses the online mechanism's performance guarantee. Section \ref{section:simulation} presents numerical results to validate the performance of the online mechanism.

\section{PROBLEM MODEL}
\label{section: prob_form}

\subsection{EVSE Reservation System and Auction}

In this section, we describe the attributes of the EVSE reservation system as well as the users of the system.  We consider parking facilities at $L$ different locations where arriving EV owners can park and charge their vehicles. The locations are dispersed and EV owners have preference to park at locations close to the desired destination for their visit. Moreover, at each location $l\in \mathcal{L}$ there are $M_{l}$ SOMC EVSEs available. Since SOMC EVSEs allow EVs to be plugged in but not receive charge, users do not have to remove their vehicle until the end of their visit. Each of the $M_l$ SOMC EVSEs at location $l$ are equipped with $C_l$ cables allowing $M_l C_l$ EVs to be parked within the location per time slot $t=1,\dots,T$. The EVSEs are constrained by the power limitations of the hardware and are limited to a maximum charge rate of $E_l$ units of power. 

In order to provide renewable energy for charging the  EVs, each parking location is equipped with a solar generation system. For each parking location $l$, the generated solar is time varying and we denote it as $s_l(t)\in[0,\overline{S}_l]$ where $\overline{S}_l$ is the maximum solar generation at location $l$. To optimize for behind-the-meter integration, facility management aims to use as much solar energy as possible before purchasing energy from the local distribution grid. We denote the price of energy purchased from the grid to serve location $l$ at time $t$ as $\pi_{l}(t)$. Each parking location $l$ can procure up to $G_l(t)$ units of energy at any time $t$ (e.g., due to a transformer limit).

The visiting EV owners arrive and depart throughout the day and request charge at various locations. There are $N$ total users participating in the EVSE auction each day and each user $n\in\mathcal{N}$ can be characterized by user `type':
\begin{equation}
\theta_n = \{ t_n^-, t_n^+, h_n, \{l_n\}, \{v_{nl}\} \} \in \Theta,
\end{equation}
where $\Theta$ is the type space of all possible users. The parameters are defined as follows. Suppose user $n$ wants to charge her EV and submits her reservation request at time $t_n$. By submitting a reservation request, user $n$ commits to arriving at any of her desired locations $\{l_n\}$ at time $t_n^-$ and leaving at $t_n^+$. During this time period, she requests that her EV receives $h_n$ %\leq t_n^+ - t_n^-$
units of energy. The last component of $\theta_n$ is $\{v_{nl}\}$, which represents the values user $n$ obtains if her EV is assigned to receive charging at location $l$. We make the assumption users receive non-negative value by charging their EVs, i.e., $v_{nl} \geq 0$. 

\subsection{Offline EVSE Reservation Problem}

When making a request to park and charge, each user $n$ submits her user type $\theta_n$. The EVSE reservation system uses the information in $\theta_n$ to generate the set of feasible assignments and charging schedules (we denote these as options $\mathcal{O}_n$) that fulfill the user's demands. Each option $o\in \mathcal{O}_n$ for user $n$, corresponds to a cable reservation and a charging schedule at an EVSE. We let $c_{no}^{ml}(t)$ denote the cable reservation request in option $o$ of user $n$ at EVSE $m$ at location $l$. We restrict $c_{no}^{ml}(t)$ to be either $0$, meaning no cable is requested at time $t$, or $1$, meaning user $n$ wants a cable reservation at $t$ in option $o$ at EVSE $m$ at location $l$. Similarly, we denote $e_{no}^{ml}(t)$ as the energy schedule for user $n$ in option $o$. The energy schedule, $e_{no}^{ml}(t)$, allows facility management to customize  when each EV will receive charge and when they will be idle at the EVSE. %The energy schedule can only take a non-zero value at time $t$ if the EV has access to a cable at time $t$. 
We allow $e_{no}^{ml}(t)$ to take different values (from a discrete set) over the usage period $t\in [t_n^-,t_n^+]$, which allows each option $o\in\mathcal{O}_n$ to request different amounts of energy at each time, as long as the user's total demands are met. 

%Additionally, each user $n$ provides the location $k\in\mathcal{L}$ where she would like to charge her EV and her willingness-to-pay denoted $v_n$. 
With this notation, the generated options for users' reservation requests (bids) can be expressed as:
\begin{equation}
B_n = \{ t_n^-, t_n^+, \{ c_{no}^{ml}(t) \}, \{ e_{no}^{ml}(t) \}, \{l_n\}, \{v_{nl}\} \}.
\end{equation}
After generating each bid package, the reservation system decides whether to accept it as well as selecting which option should fulfill the request if accepted. The binary variable $x_{no}^{ml}$ is set to 1 if option $o$ of user $n$ is accepted at EVSE $m$ at location $l$ and 0 otherwise. The reservation system also computes a payment $\hat{p}_{no}^{ml}$ for each user $n$ to pay if option $o$ in their bid package is chosen. If a user is not admitted into the EVSE reservation system, her utility is set to zero and she parks at an auxiliary parking lot without EVSEs.

The cable and energy demands at EVSE $m$ at location $l$ are denoted $y_{c}^{ml}(t)$ and $y_{e}^{ml}(t)$, respectively and are given by:
\begin{align}
	\label{eqn:offline cable demand}
	& y_{c}^{ml}(t)=\sum_{\mathcal{N},\mathcal{O}_n} c_{no}^{ml}(t) x_{no}^{ml},\\
	\label{eqn:offline energy demand}
	& y_{e}^{ml}(t)=\sum_{\mathcal{N},\mathcal{O}_n} e_{no}^{ml}(t) x_{no}^{ml}.
\end{align}
Additionally, each parking area $l$ has to generate or procure the energy needed to satisfy all the EVSE's charge schedules. As such, we denote $y_{g}^{l}(t)$ as the total energy procurement demand for location $l$. The total energy procurement demand can be calculated as follows:
\begin{align}
	\label{eqn:offline gen demand}
	& y_{g}^{l}(t)=\sum_{\mathcal{N},\mathcal{O}_n,\mathcal{M}_l} e_{no}^{ml}(t) x_{no}^{ml}.
\end{align}
% For the SOMC EVSE cables, location $l$ does not incur any cost from allocating these cables; however, there is a limited supply. We represent this as a zero-infinity cost function (cost becomes infinity when capacity is reached). The operational cost due to cables at EVSE $m$ at parking at location $l$ is denoted as
% \begin{equation}
% \label{eq:cable cost}
%     f_{c}^{ml}(y_{c}^{ml}(t)) = \begin{cases}
%         0 & y_{c}^{ml}(t) \in [0,C_l] \\
%         +\infty &y_{c}^{ml}(t) > C_l.
%         \end{cases}
% \end{equation}
% For the EVSE energy, we have
% \begin{equation}
% \label{eq:energy cost}
%     f_{e}^{ml}(y_{e}^{ml}(t)) = \begin{cases}
%         0 & y_{e}^{ml}(t) \in [0,E_l] \\
%         +\infty &y_{e}^{ml}(t) > E_l.
%         \end{cases}
% \end{equation}

Each location $l$ has an operational cost function due to the total amount of energy needed to satisfy all the admitted EVs. For the energy procurement at location $l$, we have the following operational cost function:
\begin{align}
	\label{eq:generation cost}
	&f_{g}^{l}(y_{g}^{l}(t)) =  \nonumber\\
	&\begin{cases}
		0 & y_g^{l}(t) \in [0,s_l(t)]\\
		\pi_l(t) (y_{g}^{l}(t)-s_l(t)) & y_g^{l}(t) \in \big(s_l(t),s_l(t)+G_l(t)] \\
		+\infty & y_{g}^{l}(t) > s_l(t)+G_l(t).
	\end{cases}
\end{align}
Equation \eqref{eq:generation cost} represents the cost to produce the energy needed in the whole parking location at each time slot. While the demand is less than the available solar, the operational cost is zero. Once the demand exceeds the available solar generation, energy is bought from the grid until the transformer limit $G_l(t)$ is reached. After this point, no more energy can be procured and the operational cost is set to infinity. %Figure \ref{fig: Model} shows an example for one parking location and four arriving EVs.
The goal of the EVSE reservation system is to assign and schedule users to the available EVSEs to maximize social welfare. If all the information of the $N$ requests within the time span $[0,T]$ is known in advance (assuming truthful user valuations), we can write the following offline social welfare maximization problem for assigning EVs to EVSEs and determining their charging plan:
\begin{subequations}
	\begin{align}
		\label{eqn:offline obj}
		&\max_{x} \sum_{\mathcal{N},\mathcal{O}_n,\mathcal{L},\mathcal{M}_l} v_{nl} x_{no}^{ml}-\sum_{\mathcal{T}, \mathcal{L}}    f_{g}^{l}(y_{g}^{l}(t)) \\ \nonumber
		% & \hspace{20}-\sum_{\mathcal{T}, \mathcal{L}, \mathcal{M}_l} f_{e}^{ml}(y_{e}^{ml}(t)) \\\nonumber
		% & \hspace{20}-\sum_{\mathcal{T}, \mathcal{L}}    f_{g}^{l}(y_{g}^{l}(t)) \\\nonumber
		&\textrm{ subject to:}\nonumber\\
		\label{eqn:offline one option}
		&\sum_{\mathcal{O}_n, \mathcal{L},\mathcal{M}_{l}} x_{no}^{ml} \leq 1, \quad\forall\; n\\
		&\label{eqn:offline integer}\hspace{1pt} x_{no}^{ml}  \in \{ 0,1 \}, \hspace{11pt}\forall \; n, o, l, m\\
		&\label{eqn:offline cable lim}\hspace{1pt}y_{c}^{ml}(t) \leq C_l, \hspace{13pt} \forall \;l, m, t\\
		&\label{eqn:offline energy lim}\hspace{1pt}y_{e}^{ml}(t) \leq E_l, \hspace{13pt} \forall \;l, m, t\\
		& \nonumber \textrm{ and }\eqref{eqn:offline cable demand}, \eqref{eqn:offline energy demand}, \eqref{eqn:offline gen demand}.
		%&    \label{eq: demand geq 0}y_{r}^{ml}(t)  \geq 0, \\*
		%&\nonumber\hspace{49}\forall r\in\{c,e,g\}, m\in\mathcal{M}_{l},     l\in\mathcal{L}, t\in[0,T].
	\end{align}
\end{subequations}
\noindent Here, the objective \eqref{eqn:offline obj} is to maximize the total welfare of all the users minus the operational costs. Constraint \eqref{eqn:offline one option} ensures that at most one option is selected for each user. Constraint \eqref{eqn:offline integer} is an integer constraint on the assignment variable. Constraints \eqref{eqn:offline cable lim} and \eqref{eqn:offline energy lim} ensure that the total allocation resource demands do not exceed capacities. Equations \eqref{eqn:offline cable demand}-\eqref{eqn:offline gen demand} sum up the resource demand at each EVSE $m$ and location $l$.
If the integrality constraint \eqref{eqn:offline integer} is relaxed to $x_{no}^{ml} \geq 0$ (constraint \eqref{eqn:offline one option} ensures $x_{no}^{ml}\leq 1$), we can find the Fenchel dual of \eqref{eqn:offline obj}-\eqref{eqn:offline energy lim}. We set $u_n$ and $p_{c}^{ml}(t)$, $p_{e}^{ml}(t)$, $p_{g}^{l}(t)$ as the dual variables for constraint \eqref{eqn:offline one option} and \eqref{eqn:offline cable demand}, \eqref{eqn:offline energy demand}, \eqref{eqn:offline gen demand}, respectively. In the following, the so-called Fenchel conjugate of a function $f(.)$ is defined as:
\begin{equation}
f^*(p(t)) = \sup_{y(t)\geq0} \big\{ p(t)y(t) - f(y(t)) \big\}.
\end{equation} 
The Fenchel dual of \eqref{eqn:offline obj}-\eqref{eqn:offline energy lim} can be written:
% \vspace{5}
% \setlength{\abovedisplayskip}{0pt}
% \setlength{\belowdisplayskip}{0pt}
\begin{subequations}
	\begin{alignat}{3}
		\label{eqn:dual obj}
		&\min_{u,p} \sum_{\mathcal{N}} u_n + \sum_{\mathcal{T}, \mathcal{L}}    f_{g}^{l*}(p_{g}^{l}(t))\\* 
		&\nonumber+ \sum_{\mathcal{T}, \mathcal{L},\mathcal{M}_l} \Big( f_{c}^{ml*}(p_{c}^{ml}(t)) + f_{e}^{ml*}(p_{e}^{ml}(t)) \Big)
		%& \nonumber\hspace{51}+ \sum_{\mathcal{T}, \mathcal{L},\mathcal{M}_l} f_{e}^{ml}^*(p_{e}^{ml}(t)) \\*
		\\*
		&\nonumber \textrm{ subject to:}\\*
		\label{eqn:dual user utility}
		&\;u_n \geq \;v_{nl} 
		- \sum_{\mathcal{T}} \Big( c_{no}^{ml}(t)p_{c}^{ml}(t)\\*
		&\hspace{42pt}+e_{no}^{ml}(t)\big(p_{e}^{ml}(t)+p_{g}^{l}(t)\big)\Big)\hspace{10pt}\forall\;n,o,l,m\nonumber\\*
		\label{eqn:dual ui}
		&u_n \geq \;0, \hspace{90pt}\forall \;n\\*
		\label{eqn:dual pj}
		&p_{c}^{ml}(t),\;p_{e}^{ml}(t),\;p_{g}^{l}(t) \geq \;0,  \hspace{12pt}\forall \;l,m,t,
	\end{alignat}
\end{subequations}
\noindent where $f^*(p(t))$ is the Fenchel conjugate for the limited resources' dual variables. The Fenchel conjugates for the capacity constraints are as follows:
\begin{align}
	\label{eq:fenchel cable cost}
	&f^{ml*}_{c}(p^{ml}_{c}(t)) = p^{ml}_{c}(t)C_l, \hspace{24pt} p^{ml}_{c}(t) \geq 0\\
	\label{eq:fenchel energy cost}
	&f^{ml*}_{e}(p^{ml}_{e}(t)) = p^{ml}_{e}(t)E_l, \hspace{24pt} p^{ml}_{e}(t) \geq 0.
\end{align}
The Fenchel conjugate for the energy procurement operational cost function is as follows:
% \begin{align}
% \label{eq:fenchel gen cost}
%     &f^{l*}_{g}(p^{l}_{g}(t)) =  \\*
%     &\nonumber\begin{cases}
%         s_l(t)p^{l}_{g}(t), &p^{l}_{g}(t) \leq \frac{(s_l(t)+G_l(t))\pi_l(t)}{G_l(t)} \\
%         (s_l(t)+G_l(t))(p^{l}_g(t)-\pi_l(t)) & p^{l}_{g}(t) > \frac{(s_l(t)+G_l(t))\pi_l(t)}{G_l(t)}. \\
%     \end{cases}
% \end{align}
\begin{align}
	\label{eq:fenchel gen cost}
	&f^{l*}_{g}(p^{l}_{g}(t)) =  \\*
	&\nonumber\begin{cases}
		s_l(t)p^{l}_{g}(t), &p^{l}_{g}(t) < \pi_l(t) \\
		(s_l(t)+G_l(t))p^{l}_g(t)-G_l(t)\pi_l(t) & p^{l}_{g}(t) \geq \pi_l(t). \\
	\end{cases}
\end{align}

\subsection{Admittance, Rejection, and Allocation Decisions}

In this section, we discuss how the EVSE system determines whether to accept or reject as well as how to allocate user $n$ if accepted. The EVSE reservation system assigns $x_{no}^{ml}=1$ for some option $o\in\mathcal{O}_n$ if user $n$ is accepted into the reservation system. For each user, the KKT conditions for constraint \eqref{eqn:dual user utility} in the offline dual problem indicate whether or not a user should be admitted into the system. In the offline solution, $u_n$ will be zero unless constraint \eqref{eqn:dual user utility} is tight for some $m\in\mathcal{M}_{l}$, $l\in\mathcal{L}$ and $o\in\mathcal{O}_n$. The reservation system solves the following equation to calculate user $n$'s utility:
\begin{align}
	\label{eq: u_n}
	&u_n =\max\Big\{0,\max_{\mathcal{O}_n,\mathcal{L},\mathcal{M}_l} \big\{v_{nl} 
	\\*\nonumber&- \sum_{t\in[t_n^-,t_n^+]} \big( c_{no}^{ml}(t)p_{c}^{ml}(t)+e_{no}^{ml}(t)(p_{e}^{ml}(t)+p_{g}^{l}(t) )\big)\big\}\Big\}.
\end{align}
If $u_n$ returns zero, the utility of admitting user $n$ into the system is not large enough; therefore, user $n$ is denied a reservation and is sent to auxiliary parking. If $u_n$ returns a positive value, user $n$ is admitted into the reservation system with $o\in\mathcal{O}_n$ that maximizes equation \eqref{eq: u_n}.

In equation \eqref{eq: u_n}, $p_{c}^{ml}(t)$, $p_{e}^{ml}(t)$, and $p_{g}^{l}(t)$ are the marginal prices per unit of limited resource at EVSE $m$ at location $l$. As such, the payment user $n$ must pay if admitted into the system with option $o$ can be written:
\begin{align}
	\label{eq: payment}
	\hat{p}_{no}^{ml} = \sum_{t\in[t_n^-,t_n^+]} \Big( c_{no}^{ml}(t)p_{c}^{ml}(t)+e_{no}^{ml}(t)(p_{e}^{ml}(t)+p_{g}^{l}(t)) \Big).
\end{align}
Since the system is only admitting users with positive utilities, the auction provides individual rationality for all users $n\in\mathcal{N}$. Furthermore, the system is assigning each user to the option that maximizes the user's utility function (valuation minus payment) with respect to the current marginal prices. %As such, the mechanism is targeting user truthfulness as well as welfare maximization for the users and facility management.

The offline primal and dual formulations for the EVSE reservation problem in \eqref{eqn:offline obj}-\eqref{eqn:offline energy lim} and \eqref{eqn:dual obj}-\eqref{eqn:dual pj} are established assuming complete knowledge of all $N$ users over the entire time span. However, users submit reservations at random times throughout the day, prohibiting an offline approach. For example, when user $n$ arrives, a new primal variable $x_{no}^{ml}$ and dual variable $u_n$ must be assigned while still meeting the constraints. The EVSE reservation system must decide immediately whether to admit user $n$ into the parking structure. Furthermore, if user $n$ is accepted, the EVSE reservation system decides which option will fulfill the request and calculates the user's payment, which cannot be revoked or modified later. 

In the following, we discuss an online pricing solution based on \eqref{eqn:offline obj}-\eqref{eqn:offline energy lim} and \eqref{eqn:dual obj}-\eqref{eqn:dual pj} in order to solve the EVSE reservation problem and determine users' payments in an online fashion.

\section{Online Pricing Mechanism}
\label{section: Online_Mech}

\subsection{Payment Design}

In the offline problem, the total demands $y(t)$ for each resource are known before solving. As such, the prices are calculated as follows: if the demand $y(t)$ is less than the capacity of a resource, set the price $p(t)$ equal to the marginal operational cost $f'(y(t))$. In this case, each user will pay for cost of their allocation. For both the EVSE cables and energy, since these resources do not have an operational cost, the marginal prices $p(t)$ are zero if the demand is below capacity. If the demand of a resource exceeds the capacity, the marginal prices act as filters to reject users with low valuations until the filtered demand matches the capacity.

However, the EVSE reservation problem requires an online solution. In this case, traditional dynamic programming based strategies that rely on input models fall short for deriving the optimal marginal prices due to intractable state-space size and potentially inaccurate statistics. Since the demands for the limited resources are not known in advance, the marginal prices must be calculated online. 

In the remainder of this section, we describe how the online EVSE reservation system calculates the optimal marginal prices on EVSE cables, energy, and generation based on a pricing heuristic, for which we provide performance guarantees. Specifically, our EVSE reservation system updates the prices $p(t)$ heuristically as the amounts of allocated resources $y(t)$ evolve, but only based on past observations. The pricing scheme has two major goals: (1) to make sure that the marginal gain in welfare from an allocation is greater than the operational cost incurred to serve the allocation, and (2) to filter out low value users early to ensure there are adequate resources for higher value users later on.

The structure of the pricing functions we use is adopted from \cite{IaaS}, where the authors present a pricing framework for data centers with limited computation resources and server costs  under an adverserial setting. For the EVSE cables, the proposed marginal payment function  is as follows:
\begin{align}
	\label{eqn:zero inf price}
	p_{c}^{ml}(y_{c}^{ml}(t)) =& \Big(\frac{L_c}{4\sum_{\mathcal{L}}(M_l + \frac{1}{2})}\Big) \\*
	&\nonumber \times\Big( \frac{4\sum_{\mathcal{L}}(M_l + \frac{1}{2})U_c}{L_c} \Big)^{\frac{y_{c}^{ml}(t)}{C_l}},
\end{align}
where $y_{c}^{ml}(t)$ is the current demand for the cables at EVSE $m$ at location $l$ at time $t$. Furthermore, $L_c$ and $U_c$ are respectively the lower and upper bounds on users' value per cable per unit of time, which are defined as:
\begin{subequations}
	\begin{align}
		& \label{eqn:Lc}L_c = \min_{\mathcal{N},\mathcal{O}_n,\mathcal{L},\mathcal{M}_l} \frac{v_{nl}}{2\sum_{\mathcal{L}}(M_l+\frac{1}{2})\sum_{t\in[t_n^-,t_n^+]}c_{no}^{ml}(t)},\\
		&\label{eqn:Uc}U_c = \max_{\mathcal{N},\mathcal{O}_n,\mathcal{L},\mathcal{M}_l,\mathcal{T}} \frac{v_{nl}}{c_{no}^{ml}(t)}, \quad {c_{no}^{ml}(t)} \neq 0.
	\end{align}
\end{subequations}
For the pricing function for EVSE energy units, we change $C_l$ to $E_l$ in the exponent of equation \eqref{eqn:zero inf price} and calculate $L_e$ and $U_e$ using $e_{no}^{nl}(t)$ in \eqref{eqn:Lc} and \eqref{eqn:Uc}. Additionally, $L_g$ and $U_g$ are the same as $L_e$ and $U_e$, respectively.
When $y_{c}^{ml}(t)=0$ we note that \eqref{eqn:zero inf price} outputs a price low enough that any user will be accepted (subject to $L_c$). As $y_{c}^{ml}(t)$ increases, the price increases exponentially. When $y_{c}^{ml}(t)$ is equal to the capacity, the marginal price is high enough to reject any user (because we assume $U_c$ is known beforehand). This ensures resource capacity constraints will always be upheld.

For the piecewise linear operational cost to procure energy in \eqref{eq:generation cost}, we propose the following pricing function:
\begin{align}
	\label{eqn: procureprice1}
	p_{g}^{l}
	(y_{g}^{l}(t)) = \pi_l(t) + \Big(\frac{L_g-\pi_l(t)}{4\sum_{\mathcal{L}}(M_l + \frac{1}{2})}\Big) \times\\*
	\nonumber \Big( \frac{4\sum_{\mathcal{L}}(M_l + \frac{1}{2})(U_g-\pi_l(t))}{L_g-\pi_l(t)} \Big)^{\frac{y_g^l(t)}{s_l(t)+G_l(t)}}.
\end{align}
The marginal pricing function \eqref{eqn: procureprice1} for electricity procurement is similar to the pricing function for the EVSE cables and energy; however, the price of electricity $\pi_l(t)$ is included in \eqref{eqn: procureprice1} to ensure each user's payment is greater than the electricity cost needed to charge their vehicle.

\subsection{Online Auction Mechanism}
In this section we describe the online EVSE reservation algorithm titled \textsc{OnlineEvseReservation} presented in Algorithm \ref{algorithm}. When each user $n$ arrives, the system first generates the possible charge schedule options $\mathcal{O}_n$ that fulfill her demands. The algorithm then decides whether to accept or reject user $n$ depending on user $n$'s potential utility gain due to her valuation and the current resource prices (line \ref{alg:utility}). If user $n$ is admitted and allocated option $o^{\star}\in\mathcal{O}_n$ at EVSE $m^{\star}\in\mathcal{M}_{l}$ at location $l^{\star}$, she is charged payment according to the total amount of cables, energy, and generation allocated and the current marginal prices. The algorithm updates the primal variables $x_{no}^{ml}$ after each acceptance and rejection. The total resource demands $y(t)$ are updated in line \ref{alg:demand update} if user $n$ is accepted into the system. Similarly, the marginal resource prices $p(t)$ are updated accordingly in line \ref{alg:price update}.

\begin{algorithm}[]
	\small
	\caption{\textsc{OnlineEvseReservation}}
	\label{algorithm}
	\begin{algorithmic}
		\STATE \textbf{Input:} $\mathcal{L}, \mathcal{M}_l,C_l, E_l, G_l, S_l, \pi_l, L_{c,e,g}, U_{c,e,g}$
		\STATE \textbf{Output:} $x, p$
	\end{algorithmic}
	\begin{algorithmic}[1]
		\STATE Define $f_g^l(y_g^l(t))$ according to \eqref{eq:generation cost} at all locations.
		\STATE Define the pricing functions $p(y(t))$ according to \eqref{eqn:zero inf price} and \eqref{eqn: procureprice1} for cables, energy, and generation at all EVSEs and locations.
		\STATE Initialize $x_{no}^{ml}=0$, $y^{ml}(t)=0$, $u_n=0$.
		\STATE Initialize prices $p(0)$ according to \eqref{eqn:zero inf price} and \eqref{eqn: procureprice1}.
		\STATE \textbf{Repeat for all $N$ users:}
		\STATE User $n$ submits $\theta_n$, generate feasible charging options $B_n$.
		\STATE Update dual variable $u_n$ according to \eqref{eq: u_n}. \label{alg:utility}
		\IF{$u_n > 0$}
		\STATE $(o^{\star},m^{\star},l^{\star}) =\argmax_{\mathcal{L}, \mathcal{M}_{l},\mathcal{O}_n}\big\{v_{nl}$ \\
		\vspace{3pt}
		\hspace{20pt} $- \sum_{t\in[t_n^-,t_n^+]} \big( c_{no}^{ml}(t)p_{c}^{ml}(t)$\\
		\vspace{3pt}
		\hspace{20pt} $ +e_{no}^{ml}(t)(p_{e}^{ml}(t)+p_{g}^{l}(t)) \big)\big\}$
		\vspace{3pt}
		\STATE $\hat{p}_{no^{\star}}^{m^{\star}l^{\star}} = \sum_{t\in[t_n^-,t_n^+]} \Big( c_{no^{\star}}^{m^{\star}l^{\star}}(t)p_{c}^{m^{\star}l^{\star}}(t)$\\*
		\hspace{20pt} $+e_{no^{\star}}^{m^{\star}l^{\star}}(t)(p_{e}^{m^{\star}l^{\star}}(t) +p_{g}^{l^{\star}}(t)) \Big)$
		\vspace{3pt}
		\STATE $x_{no^{\star}}^{m^{\star}l^{\star}}=1$ and $x_{no}^{ml}=0$ for all $(o,l,m) \neq (o^{\star},l^{\star},m^{\star})$
		\STATE Update total demand $y(t)$ for cables, energy, and generation according to \eqref{eqn:offline cable demand}-\eqref{eqn:offline gen demand}. \label{alg:demand update}
		\STATE Update marginal prices $p(t)$ for cables, energy, and generation according to \eqref{eqn:zero inf price} and \eqref{eqn: procureprice1}. \label{alg:price update}
		\ELSE
		\STATE $x_{no}^{ml}=0$, \hspace{3pt}  $\forall$ $l$, $m$, $o$.
		\ENDIF 
		\IF{$\exists o^{\star},m^{\star},l^{\star}$ and $x_{no^{\star}}^{m^{\star}l^{\star}}=1$}
		\STATE Accept user $n$ and allocate cables and energy in parking location $l^{\star}$ at EVSE $m^{\star}$.
		\STATE Charge user $n$ at $\hat{p}_{no^{\star}}^{m^{\star}l^{\star}}$.
		\ELSE
		\STATE Send user $n$ to auxiliary parking.
		\ENDIF
	\end{algorithmic}
\end{algorithm}

We compare the total social welfare resulting from the online solution to the optimal offline solution. Specifically, an online mechanism is said to be $\alpha$-competitive when the ratio of social welfare from the optimal offline solution to the social welfare from the mechanism is bounded by $\alpha$. We extend the competitive ratio result from \cite{IaaS} in Proposition \ref{piecewise linear price}. We note that the analysis used for this competitive ratio assumes each user's resource demands $c_{no}^{ml}(t)$ and $e_{no}^{nl}(t)$ are much smaller than the capacity limits, $C_l$ and $E_l$, respectively. This ensures no user purchases too large of a fraction of the total available resources.
\begin{proposition}
	\label{piecewise linear price}
	The marginal pricing function \eqref{eqn: procureprice1} is $\alpha_1$-competitive in social welfare when selling limited resources with the piecewise linear operational cost in \eqref{eq:generation cost} where
	\begin{align}
		\nonumber
		\alpha_1= 2\max_{\mathcal{L},\mathcal{T}}\Big\{\ln{\Big(\frac{4\sum_{\mathcal{L}}(M_l + \frac{1}{2})(U_g-\pi_l(t))}{L_g-\pi_l(t)}\Big)}\Big\}.
	\end{align}
\end{proposition}
\begin{proof}
	In \cite{IaaS}, the authors show their pricing functions are $\alpha$-competitive in social welfare with respect to the buying and selling of limited computation resources at data centers. Specifically, the pricing functions, operational cost functions, and Fenchel conjugates for the limited resources need to satisfy the \textit{Differential Allocation-Payment Relationship} given by:
	\begin{align}
		\label{eqn:diffalloc}
		\big(p_g^l(t) - f_g^{l'}(y_g^l(t))\big) \text{d}y_{g}^{l}(t) \geq \frac{1}{\alpha_g^l(t)} f_g^{l*'}(p_g^l(t)) \text{d}p_g^l(t)
	\end{align}
	for all $l\in\mathcal{L}, t\in[0,T]$. The derivatives of the energy-procurement operational cost in \eqref{eq:generation cost} and its Fenchel conjugate \eqref{eq:fenchel gen cost} are $f_g^{l'}(y_g^l(t))$ and $f_g^{l*'}(p_g^l(t))$, respectively. %For the energy-procurement operational cost in \eqref{eq:generation cost} and its Fenchel conjugate \eqref{eq:fenchel gen cost} respectively, the following derivatives are:
	% \begin{align}
	%     &\label{f_deriv}f_g^{l}'(y_g^l(t)) = 
	%     \begin{cases}
	%         0, & y_g^l(t)\in [0, s_l(t)]\\
	%         \pi_l(t), & y_g^l(t)\in (s_l(t), s_l(t)+G_l(t)]
	%     \end{cases}\\
	%     &\nonumber\textrm{ and }\\
	%     &\label{f_deriv_conj}f_g^{l*}'(p_g^l(t)) = 
	%     \begin{cases}
	%         s_l(t), & p_g^l(t)\in[0,\frac{s_l(t)+G_l(t)}{G_l(t)}\pi_l(t)]\\
	%         s_l(t)+G_l(t), & p_g^l(t) > \frac{s_l(t)+G_l(t)}{G_l(t)}\pi_l(t).
	%     \end{cases}
	% \end{align}
	% The derivative of the proposed pricing function \eqref{eqn: procureprice1} is:
	% \begin{align}
	% \label{deriv accurate}
	%     \nonumber\text{d}p_g^l(t) = &\Big(\frac{L_g-\pi_l(t)}{2R(s_l(t)+G_l(t))}\Big)\Big(\frac{2R(U_g-\pi_l(t))}{L_g-\pi_l(t)}\Big)^{\frac{y_g^l(t)}{s_l(t)+G_l(t)}}\\*
	%     &\times\ln{\Big(\frac{2R(U_g-\pi_l(t))}{L_g-\pi_l(t)}\Big)}\text{d}y_{g}^{l}(t),
	% \end{align}
	% where $R=2\sum_{\mathcal{L}}M_l(C_l+E_l+\frac{1}{M_l})$. 
	Taking the derivative of the proposed pricing function \eqref{eqn: procureprice1} and setting $f_g^{l'}(y_g^l(t)) = \pi_l(t)$ minimizes the LHS of \eqref{eqn:diffalloc} and $f_g^{l*'}(p_g^l(t)) = s_l(t)+G_l(t)$ maximizes the RHS. As such, after inserting the derivative of \eqref{eqn: procureprice1} in \eqref{eqn:diffalloc}, we can show that the Differential Allocation-Payment Relationship holds with equality when choosing $\alpha_g^{l}(t)= \ln{\Big(\frac{4\sum_{\mathcal{L}}(M_l + \frac{1}{2})(U_g-\pi_l(t))}{L_g-\pi_l(t)}\Big)}$. Because \eqref{eqn:diffalloc} holds for the  pricing function, operational cost function, and Fenchel conjugate, the remainder of the proof follows from Lemma 1 and Theorem 2 in \cite{IaaS}.
\end{proof}
We note that the proposed pricing function \eqref{eqn: procureprice1} relies on accurate day-ahead forecasts of the solar generation $s_l(t)$ for $t=1,\dots,T$ at all locations $l\in\mathcal{L}$. If the daily forecasts for solar generation are inaccurate, there are two potential undesirable outcomes: 1) solar generation is overestimated and resources are over-allocated resulting in infeasible solutions, which our online solution should avoid at all costs; and 2) solar generation is underestimated and prices are set too high and the system rejects users that should otherwise be accepted. We analyze the case where we have a forecast of the solar generation each day in terms of a confidence interval. Specifically, the solar forecast takes the following form:
\begin{align}
	s_l(t) \in [\underline{s}_l(t) , \overline{s}_l(t)],  \hspace{10pt} \forall t=1,\dots,T,
\end{align}
where $s_l(t)$ is the actual solar generation at time $t$ and the terms $\underline{s}_l(t)$ and $\overline{s}_l(t)$ are lower and upper bounds given by the forecast, respectively. 
%We use the variable $\delta\in[0,1]$ to describe the precision of the forecasted interval as shown in the following equation:
% \begin{align}
%     \delta \geq \frac{\overline{s}_l(t) - \underline{s}_l(t)}{\overline{S}_l}, \hspace{10} \forall t=1,\dots,T.
% \end{align}
% Here, $\delta=0$ corresponds to a completely accurate forecast and $\delta =1$ is the same as not having a forecast at all.
To avoid possible infeasible allocations associated with overestimation of solar availability, we analyze the performance of pricing function \eqref{eqn: procureprice1} that conservatively uses the underestimate of the solar generation, $\underline{s}_l(t)$, in Proposition \ref{piecewise linear price estimate}.

\begin{proposition}
	\label{piecewise linear price estimate}
	The marginal pricing function \eqref{eqn: procureprice1} with an underestimate of solar generation, $\underline{s}_l(t)$, is $\alpha_2$-competitive in social welfare when selling limited resources with the operational cost in \eqref{eq:generation cost}
	where
	\begin{align}
		\nonumber
		\alpha_2= 2\max_{\mathcal{L},\mathcal{T}}\Big\{&\Big(\frac{\overline{s}_l(t)+G_l(t)}{\underline{s}_l(t)+G_l(t)}\Big)\times\\*
		&\nonumber\ln{\Big(\frac{4\sum_{\mathcal{L}}(M_l + \frac{1}{2})(U_g-\pi_l(t))}{L_g-\pi_l(t)}\Big)}\Big\}.
	\end{align}
	\begin{proof}
		Similar to Proposition \ref{piecewise linear price}, we show the pricing function, operational cost function, and Fenchel conjugate for the limited resource satisfy the Differential Allocation-Payment Relationship in \eqref{eqn:diffalloc} with underestimated solar generation amounts $\underline{s}_l(t)$.
		The derivatives of the energy-procurement operational cost in \eqref{eq:generation cost} and its Fenchel conjugate \eqref{eq:fenchel gen cost} remain the same.
		% The derivative of the proposed pricing function \eqref{eqn: procureprice1} with underestimated solar $\underline{s}_l(t)$ is given by:
		% \begin{align}
		% \label{deriv_estimate}
		%     &\text{d}p_g^l(t) = \Big(\frac{L_g-\pi_l(t)}{2R(\underline{s}_l(t) +G_l(t))}\Big)\times\\*
		%     &\nonumber\Big(\frac{2R(U_g-\pi_l(t))}{L_g-\pi_l(t)}\Big)^{\frac{y_g^l(t)}{\underline{s}_l(t)+G_l(t)}}
		%     \ln{\Big(\frac{2R(U_g-\pi_l(t))}{L_g-\pi_l(t)}\Big)}\text{d}y_{g}^{l}(t),
		% \end{align}
		% where $R=2\sum_{\mathcal{L}}M_l(C_l+E_l+\frac{1}{M_l})$. 
		Taking the derivative of the proposed pricing function \eqref{eqn: procureprice1} with underestimated solar $\underline{s}_l(t)$ and setting $f_g^{l'}(y_g^l(t)) = \pi_l(t)$ minimizes the LHS of \eqref{eqn:diffalloc} and $f_g^{l*'}(p_g^l(t)) = s_l(t)+G_l(t)$ maximizes the RHS. As such, after inserting the derivative of \eqref{eqn: procureprice1} in \eqref{eqn:diffalloc}, we can show that the Differential Allocation-Payment Relationship holds when $\alpha_g^{l}(t)=\big(\frac{\overline{s}_l(t)+G_l(t)}{\underline{s}_l(t)+G_l(t)}\big) \ln{\Big(\frac{4\sum_{\mathcal{L}}(M_l + \frac{1}{2})(U_g-\pi_l(t))}{L_g-\pi_l(t)}\Big)}$. Because \eqref{eqn:diffalloc} holds for the  pricing function with underestimated solar generation, operational cost function, and Fenchel conjugate, the remainder of the proof follows from Lemma 1 and Theorem 2 in \cite{IaaS}.
	\end{proof}
\end{proposition}

\section{Experimental Evaluation}
\label{section:simulation}
In this section, we present simulation results highlighting the performance of the EVSE reservation system. Electricity prices and solar generation data (see Figure \ref{fig: LMP}) were sourced from actual California ISO data in the Bay Area \cite{LMPweb},\cite{Solar}. 
\begin{figure}[]
	\centering
	\includegraphics[width=1.0\columnwidth]{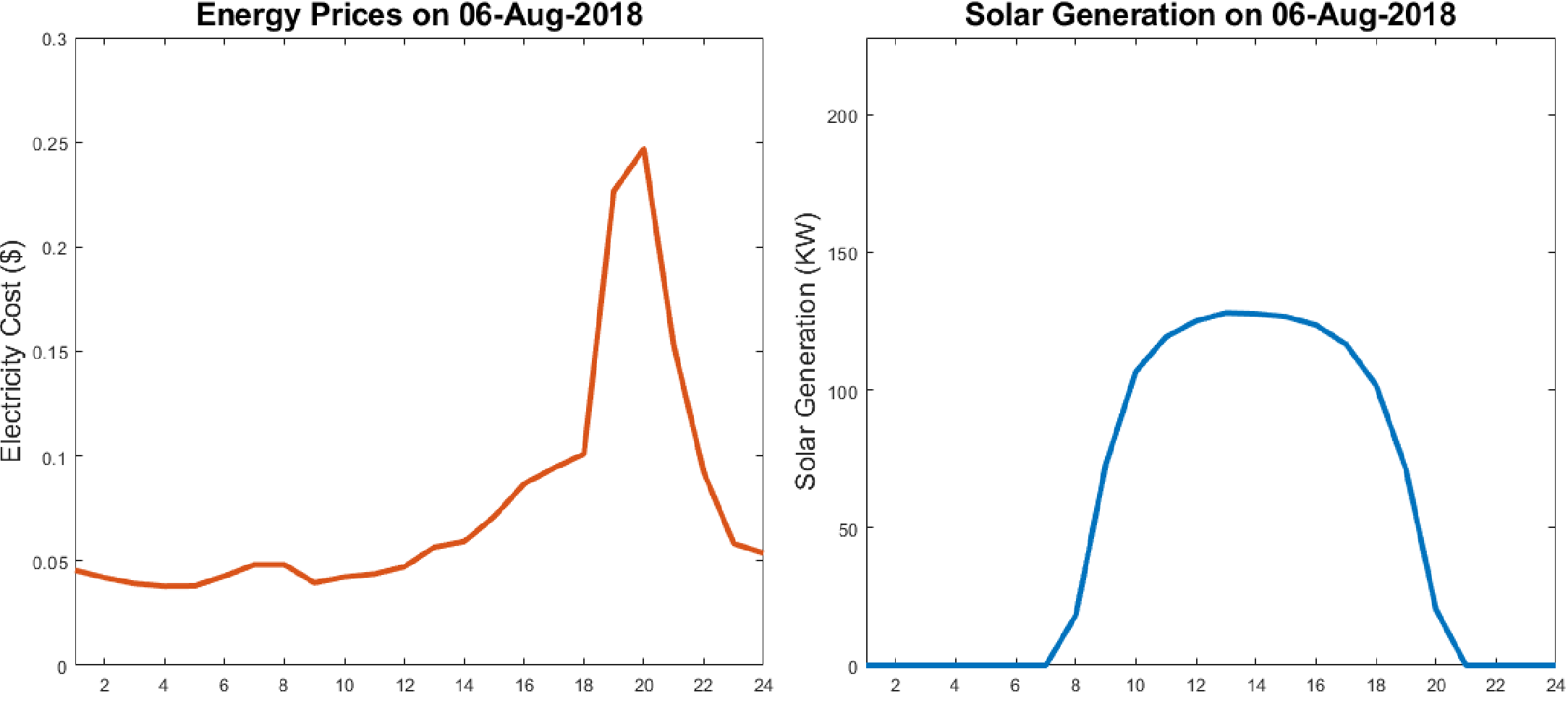}
	\caption{Left: Bay Area electricity prices. Right: Solar generation profile.}
	\label{fig: LMP}
\end{figure}
We simulated for a populated downtown area with $l=9$ different parking locations where users can park and charge their EVs. The number of EVSEs and cables available at each location are listed in Table \ref{fig: table}.
\begin{table}[]
	\small
	\begin{center}
		\begin{tabular}{||c c c c c c||} 
			\hline
			Location & $M_l$ & $C_l$ & $\substack{\text{Peak}\\\text{Generation Price}}$ & $\substack{\text{Peak}\\\text{Cable Price}}$ & $\substack{\text{Number of}\\\text{EVs Served}}$\\ [0.5ex] 
			\hline\hline
			1 & 4 & 4 & 0.247 &7.4508 &58\\ 
			\hline
			2 & 4 & 4 & 0.247 &7.4508 &67\\
			\hline
			3 & 8 & 4 & 0.730 &7.4508 &139\\
			\hline
			4 & 8 & 4 & 0.718 &7.4508 &131\\
			\hline
			5 & 2 & 4 & 0.247 &7.4508 &32\\
			\hline
			6 & 8 & 4 & 0.884 &7.4508 &128\\
			\hline
			7 & 2 & 4 & 0.247 &7.4508 &28\\
			\hline
			8 & 4 & 4 & 0.247 &7.4508 &74\\
			\hline
			9 & 2 & 4 & 0.247 &7.4508 &32\\ [0.0ex] 
			\hline
		\end{tabular}
	\end{center}
	\caption{Columns 2-3: EVSE and cable counts. Columns 4-6: Online mechanism results.}
	\label{fig: table}
\end{table}
All nine locations in the downtown area make use of the same solar generation system with maximum generation of $512$ kWh per time unit. Similarly, the area can procure energy from the grid, with a total procurement limit $G_{\mathcal{L}}(t)=512$ kWh per time unit. We simulated with $N=1000$ users with various arrival and departure times throughout the day. %displayed in Figure \ref{fig: hist}.
% \begin{figure}[h]
%     \includegraphics[width=1.\columnwidth]{images/histograms.PNG}
%     \caption{Left: Arrival time histogram. Right: Departure time histogram}
%     \label{fig: hist}
% \end{figure}
Each user arrives with three preferred parking locations with three different valuations $v_{nl}$. For the charge schedules $e_{no}^{ml}(t)$, each user was restricted to $0$ or $1$ kWh per time slot for the duration of time in the parking location. The maximum duration for a charge request was set to 8 hours and users valuations were in the interval $[\$1.50,\$7.50]$ depending on the amount of desired energy. 

In Figure \ref{fig: results} we compare the social welfare at each of the nine parking locations from the EVSE reservation system to the social welfare resulting from the case with no assignment mechanism as well as an upper bound on the optimal solution. For the no mechanism case, users arrive to the downtown area and choose the available assignment and charge schedule that maximizes their utility. For the upper bound on the optimal solution, cable and energy capacities at each EVSE were relaxed to prohibit any users from being sent to the auxiliary parking resulting in the maximum possible social welfare as long as each user's valuation was larger than the cost to serve them. We can see that the EVSE reservation system outperforms the no mechanism case by the largest amount in parking locations 3, 4, and 6. These are the locations that were most desired by the users; therefore, these locations had the most congested resources and the auction mechanism was able to filter out low value users.
\begin{figure}[t]
	\centering
	\includegraphics[width=1.0\columnwidth]{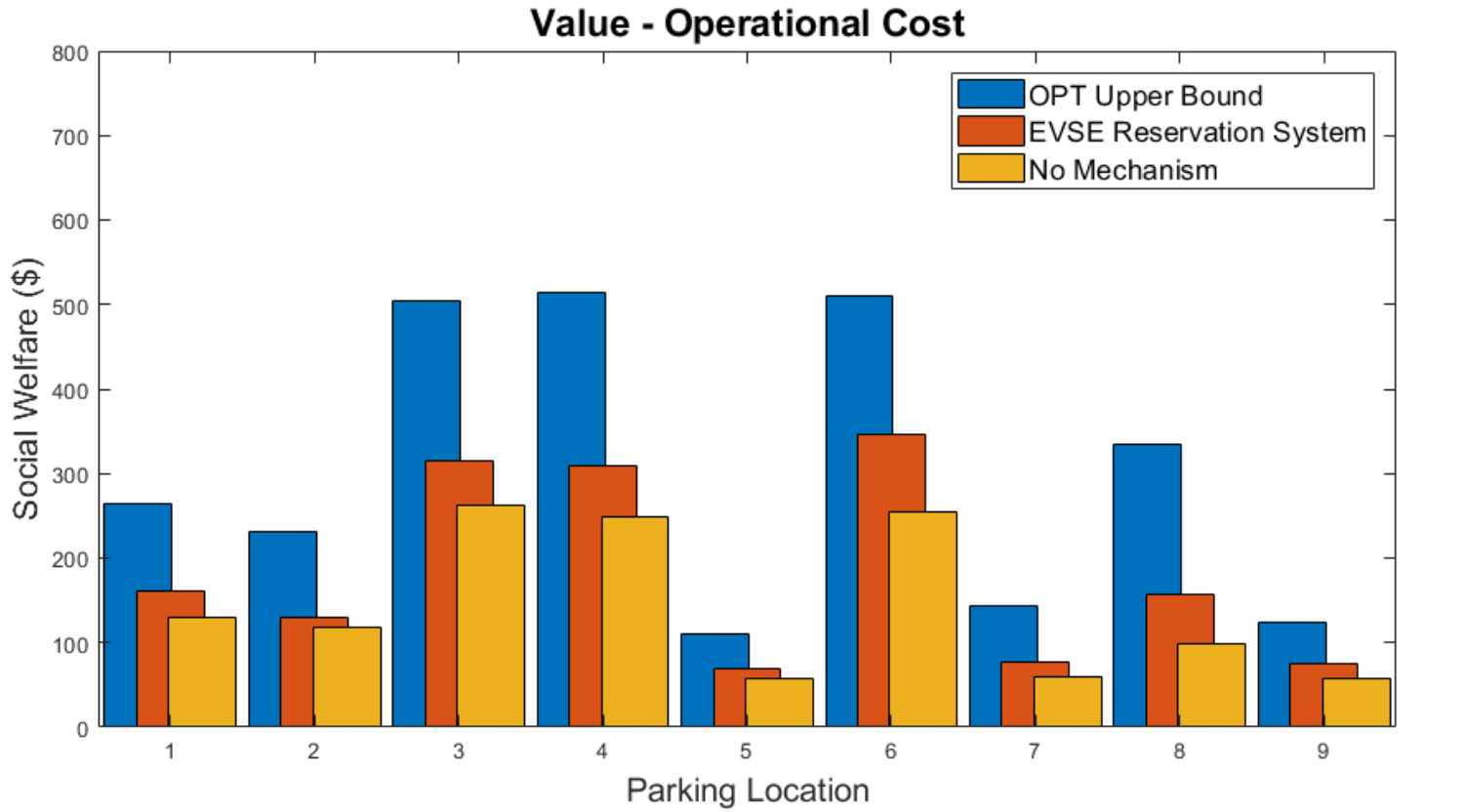}
	\caption{Social welfare comparison.}
	\label{fig: results}
\end{figure}

\section{Conclusion}
\label{section: conclusion}
In this paper, we presented a framework for the online EVSE reservation problem. Our formulation can accommodate multiple parking locations, limited resource capacities, operational costs, as well as variable arrival patterns. With this problem formulation, the parking facility management can optimize for behind-the-meter solar integration and reduce costs due to procuring electricity from the grid. We utilized an online auction mechanism to approximate the EVSE reservation problem's solution and we analyzed the performance compared to the offline solution. We provided a numerical simulation to validate the performance of the EVSE reservation system in a downtown district with multiple parking locations equipped with EVSEs. In future work, we will analyze the effects of varying levels of infrastructure investments for various locations. Additionally, we will study other smart charging benefits such as frequency regulation services or participation in demand response for parking infrastructure utilizing smart assignment mechanisms.
%\vspace{20}
% With accurate solar estimates $s_l(t)$, Algorithm 1 for the online EVSE reservation problem is $\alpha$-competitive in social welfare where $\alpha$ is defined as follows:
% \begin{align}
%     &\alpha = 2\max_{\mathcal{L},\mathcal{T}} \Big\{ \ln{ \big( \frac{2\sum_{\mathcal{L}}M_l(C_l+E_l+\frac{1}{M_l})(U_g-\pi_l(t))}{L_g-\pi_l(t)} \big) } \Big\}.
% \end{align}
% With underestimates of solar generation $\underline{s}_l(t) = (1-\delta)s_l(t), \; \forall \delta\in[0,1], \; t\in[0,T]$, Algorithm 1 for the online EVSE reservation problem is $\big(\frac{\overline{S}_l+G_l(t)}{G_l(t)}\big)\alpha$-competitive in social welfare.

%use section* for acknowledgment
%\section*{Acknowledgment}

\bibliographystyle{IEEEtran}
\bibliography{references}

\end{document}